\pdfoutput=1
\documentclass[
aps,
prl,
footinbib,
reprint,
superscriptaddress,
longbibliography]{revtex4-2}

\usepackage{amsmath,amssymb,amsfonts}
\usepackage{mathtools}    
\usepackage{bm}
\usepackage{physics2}
\usephysicsmodule{braket}
\usepackage{stmaryrd}
\usepackage{mathrsfs}
\usepackage{derivative}

\usepackage{enumitem}

\usepackage{amsthm}
\theoremstyle{definition}
\newtheorem{thm}{Theorem}

\newcommand{\R}{\mathbb{R}}
\newcommand{\C}{\mathbb{C}}

\newcommand{\id}{\mathrm{id}}

\newcommand{\sto}{\shortrightarrow}

\DeclareMathOperator{\Tr}{Tr}

\DeclarePairedDelimiter{\abs}{\lvert}{\rvert}
\DeclarePairedDelimiter{\norm}{\lVert}{\rVert}

\renewcommand{\Re}{\operatorname{Re}}

\renewcommand\paragraph[1]{%
  \par\emph{#1---}\kern2pt\relax\ignorespaces}

\usepackage[resetlabels]{multibib}
\newcites{SM}{References}
\usepackage{textcase}

\makeatletter
\AtBeginDocument{%
  \let\thebibliography\rtx@thebibliography
  \let\endthebibliography\endrtx@thebibliography
}
\makeatother

\makeatletter

\def\REV@fm@prefix{frontmatter.}
\def\REV@fm@dest#1{\REV@fm@prefix#1}

\def\frontmatter@footnotemark#1{%
  \leavevmode
  \ifhmode\edef\@x@sf{\the\spacefactor}\nobreak\fi
  \begingroup
    \hyper@linkstart{link}{\REV@fm@dest{#1}}%
      \csname c@\@mpfn\endcsname#1\relax
      \def\@thefnmark{\frontmatter@thefootnote}%
      \@makefnmark
    \hyper@linkend
  \endgroup
  \ifhmode\spacefactor\@x@sf\fi
  \relax
}

\def\present@bibnote#1#2{%
  \item[%
    \textsuperscript{%
      \normalfont
      \Hy@raisedlink{%
        \hyper@anchorstart{\REV@fm@dest{#1}}\hyper@anchorend
      }%
      \begingroup
        \csname c@\@mpfn\endcsname#1\relax
        \frontmatter@thefootnote
      \endgroup
    }%
  ]#2\par
}

\def\present@FM@footnote#1#2{%
  \begingroup
    \csname c@\@mpfn\endcsname#1\relax
    \def\@thefnmark{\frontmatter@thefootnote}%
    \Hy@raisedlink{%
      \hyper@anchorstart{\REV@fm@dest{#1}}\hyper@anchorend
    }%
    \frontmatter@footnotetext{#2}%
  \endgroup
}

\newcounter{REV@maketitle@count}

\def\frontmatter@maketitle{%
  \stepcounter{REV@maketitle@count}%
  \begingroup
    \ifnum\value{REV@maketitle@count}>1\relax
      \def\REV@fm@prefix{SMfrontmatter.}%
      \let\frontmatter@footnote@produce
          \frontmatter@footnote@produce@footnote
    \fi
    \let\REV@orig@label\label
    \def\label##1{%
      \edef\REV@tmp{##1}%
      \edef\REV@tgt{FirstPage}%
      \ifx\REV@tmp\REV@tgt
      \else
        \REV@orig@label{##1}%
      \fi
    }%
    \@author@finish
    \title@column\titleblock@produce
    \suppressfloats[t]%
    \let\abstract\@undefined
    \let\endabstract\@undefined
    \titlepage@sw{\vfil\clearpage}{}%
    \onecolumn@grid@setup
    \def\set@footnotewidth{\set@footnotewidth@one}%
  \endgroup
}

\makeatother

\usepackage{hyperref}
\hypersetup{
  colorlinks=true,
  citecolor=magenta,
  linkcolor=blue,
  urlcolor=cyan
}

\begin{document}
\newcommand{\nameoftitle}{%
Quantitative Wigner-Araki-Yanase Theorems for Unitary and Antiunitary Symmetries}

\title{\nameoftitle}

\author{Akihiro Hokkyo}
\email{hokkyo@cat.phys.s.u-tokyo.ac.jp}
\affiliation{Department of Physics, Graduate School of Science, The University of Tokyo, 7-3-1 Hongo, Bunkyo, Tokyo 113-8654, Japan}
\author{Hiroyasu Tajima}
\email{hiroyasu.tajima@inf.kyushu-u.ac.jp}
\affiliation{Information Science and Electrical Engineering, Department of Informatics, Kyushu University, Nishi-ku, Fukuoka 819-0395, Japan}
\affiliation{JST FOREST, 4-1-8 Honcho, Kawaguchi, Saitama 332-0012, Japan}

\begin{abstract}
Symmetry imposes fundamental constraints on quantum measurement and control. 
The Wigner-Araki-Yanase theorem and its quantitative extensions capture this restriction for continuous symmetries, in terms of fluctuations of conserved generators.
Such generator-based bounds, however, do not provide quantitative limitations for discrete unitary symmetries or for antiunitary symmetries.
Here we establish quantitative WAY-type theorems for symmetry-breaking projective measurements and unitary gates under arbitrary unitary and antiunitary symmetries. 
Our approach is based on a two-target no-programming inequality: 
if a single processor approximately implements two operations that amplify distinguishability, 
then the corresponding program states must themselves be distinguishable.
Applied to symmetric implementations, this converts the error of an asymmetric measurement or gate directly into a lower bound on the asymmetry of the apparatus state, quantified by its fidelity with its symmetry-transformed copy.
Our results apply to discrete and antiunitary symmetries, 
thereby providing a fundamental limit for symmetry-limited quantum measurement and control beyond the continuous-symmetry regime.
\end{abstract}

\maketitle
\paragraph{Introduction}
Symmetry is one of the most fundamental restrictions on quantum dynamics. 
If all available operations respect a symmetry, then a measurement or gate that breaks that symmetry cannot be implemented for free. 
The Wigner-Araki-Yanase (WAY) theorem~\cite{wignerMessungQuantenmechanischerOperatoren1952,arakiMeasurementQuantumMechanical1960} is the paradigmatic example of this principle. 
It states that, under an additive conservation law, an observable that does not commute with the conserved quantity cannot be measured exactly. 

Quantitative WAY theorems refine this no-go statement into tradeoff relations~\cite{yanaseOptimalMeasuringApparatus1961,
ghirardiLimitationsQuantumMeasurements1981,
ozawaConservationLawsUncertainty2002,
miyaderaWignerArakiYanaseTheoremDistinguishability2006,
buschPositionMeasurementsObeying2011,
loveridgeMeasurementQuantumMechanical2011,
korzekwaResourceTheoryAsymmetry,
tajimaCoherencevarianceUncertaintyRelation2019,
mohammadyMeasurementDisturbanceConservation2023,
emoriErrorDisturbanceIrreversibility2023}.
In particular, Yanase's analysis and subsequent refinements show that high-accuracy approximate measurements require a large fluctuation of the conserved quantity in the apparatus
~\cite{yanaseOptimalMeasuringApparatus1961,
ghirardiLimitationsQuantumMeasurements1981,
ozawaConservationLawsUncertainty2002,
buschPositionMeasurementsObeying2011,
loveridgeMeasurementQuantumMechanical2011,
korzekwaResourceTheoryAsymmetry,
tajimaCoherencevarianceUncertaintyRelation2019,
emoriErrorDisturbanceIrreversibility2023}. 
Analogous restrictions apply to quantum gates, including the CNOT gate~\cite{ozawaConservativeQuantumComputing2002},
single-qubit gates~\cite{karasawaConservationlawinducedQuantumLimits2007,karasawaGateFidelityArbitrary2009}, 
and general unitary gates~\cite{tajimaUncertaintyRelationsImplementation2018,tajimaCoherenceCostViolating2020}. 
More recently, WAY-type phenomena have been recognized as reflecting a general resource-cost principle in symmetry-constrained quantum dynamics~\cite{tajimaUniversalLimitationQuantum2022,
tajimaUniversalTradeoffStructure2022}.
It connects 
the above measurement limitations and unitary-gate limitations
with covariant quantum error correction and the Eastin--Knill theorem~\cite{eastinRestrictionsTransversalEncoded2009,faistContinuousSymmetriesApproximate2020,kubicaUsingQuantumMetrological2021}, 
and has led to bounds on energetic coherence cost of Gibbs-preserving operations in quantum thermodynamics~\cite{tajimaGibbsPreservingOperationsRequiring2025}.
Thus, the WAY theorem is not merely a constraint on quantum measurements, but a paradigmatic manifestation of a broader principle: symmetry imposes quantitative limitations on quantum dynamics.

Despite this progress, 
a quantitative WAY theory for general symmetries is still missing. 
The main obstruction is that existing quantitative bounds are formulated in terms of symmetry generators and their variances.
This generator-based approach is powerful for continuous symmetries, but it does not directly apply to discrete symmetries, such as lattice translations, or to antiunitary symmetries.
Recent resource-theoretic formulations have placed WAY-type constraints in a broader framework beyond symmetry-based settings~\cite{tajimaUniversalTradeoffRelations2025}, 
but they do not directly yield explicit quantitative bounds for discrete symmetries, 
since their continuity-based arguments for additive resource monotones break down for discrete symmetries.

In this Letter, we establish quantitative WAY-type limitations under \emph{arbitrary unitary and antiunitary symmetries}, including discrete groups and symmetries with unbounded generators (Fig.~\ref{fig:1}). 
The resulting bounds do not rely on symmetry generators or their variances. Instead, they quantify the required asymmetry of the apparatus by the fidelity between the apparatus state and its symmetry-transformed copy.
Thus, the generator variance appearing in conventional quantitative WAY theorems is replaced by a fidelity-based asymmetry measure that remains meaningful for discrete and antiunitary symmetries.
They also apply to infinite-dimensional symmetries with unbounded generators, recovering and strengthening the recent WAY theorem for unbounded conserved quantities in that setting~\cite{kuramochiWignerArakiYanaseTheoremContinuous2023}. 
\begin{figure}[tbp]
    \centering
    \includegraphics[width=\linewidth]{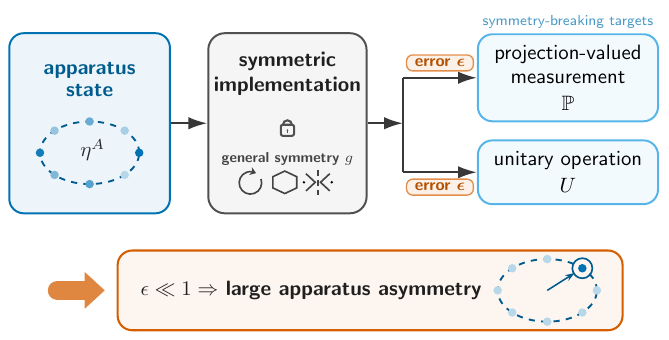}
    \caption{Symmetry imposes a quantitative cost on implementing symmetry-breaking quantum operations.
A symmetric implementation, assisted by an apparatus state $\eta^A$, is required to approximate a target operation that breaks a general symmetry $g$, such as a projection-valued measurement $\mathbb P$ or a unitary operation $U$.
For arbitrary unitary or antiunitary symmetries, 
small implementation error $\epsilon$ is possible only if the apparatus state carries sufficient asymmetry.
This tradeoff gives a quantitative WAY-type limitation without relying on symmetry generators or their variances.
}
    \label{fig:1}
\end{figure}

The key mechanism is the no-programming perspective developed by Marvian and Spekkens~\cite{marvianInformationtheoreticAccountWignerArakiYanase2012};
see also Ref.~\cite{ahmadiWignerArakiYanase2013} for a related formulation. 
In the WAY setting, applying a symmetry transformation to a symmetric implementation of an asymmetric operation with apparatus state produces another implementation, by the same processor, of the transformed operation with the transformed apparatus state. Thus the two apparatus states play the role of programs for two symmetry-related target operations, and the no-programming theorem~\cite{nielsenProgrammableQuantumGate1997,
dusekQuantumcontrolledMeasurementDevice2002} yields the qualitative WAY theorem~\cite{marvianInformationtheoreticAccountWignerArakiYanase2012}. 
This connection suggests that a quantitative WAY theorem should follow from a quantitative \emph{two-target} no-programming principle. However, a two-target quantitative no-programming bound for measurements has been missing. Although quantitative bounds have been studied for approximate programming of unitary operations~\cite{vidalStorageQuantumDynamics2000,hilleryApproximateProgrammableQuantum2006,majenzEntropyQuantumInformation2017}, the common two-target structure behind projective measurements and unitary gates had not been formulated as a unified principle. This problem is also distinct from approximate \emph{universal} programming~\cite{perez-garciaOptimalityProgrammableQuantum2006,majenzEntropyQuantumInformation2017,kubickiResourceQuantificationNoPrograming2019,yangOptimalUniversalProgramming2020}, 
which concerns large families of operations and primarily constrains the required program dimension.

We prove this missing no-programming relation for measurements and channels in a unified way.
Our no-programming inequality applies to pairs of measurements or channels that amplify distinguishability, namely, 
pairs that map some nonorthogonal inputs to perfectly distinguishable outputs.
This operational class includes, in particular, distinct projective measurements and distinct unitary gates. 
For any such pair, if the same processor implements both target operations with error at most $\epsilon$, then the fidelity between the corresponding program states must be at most $O(\epsilon)$, with a denominator determined by the distinguishability amplification of the target pair. 
Applied to symmetric implementations, 
this converts the implementation error of a symmetry-breaking measurement or gate directly into a lower bound on the asymmetry of the apparatus state. 
Our results thereby provide a general route to quantitative symmetry-limited bounds on quantum measurement and control beyond the generator-based regime.

\paragraph{Implementations of measurements and channels}
We begin with the implementation of a measurement.
A measurement $\mathbb{P}$ on $S$ with outcomes in a (measurable) space $\Omega$ is an affine map from the set of density operators on $S$ to the set of probability measures on $\Omega$, 
which assigns to each event $E\subset\Omega$ a positive operator $\mathbb{P}_E$ such that
$\mathbb{P}(\rho)(E)=\Tr(\mathbb{P}_E\rho)$.
If $\mathbb{P}_E$ is a projection for every event $E$, then $\mathbb{P}$ is called a projection-valued measure (PVM).
We measure the error between two measurements by
\begin{equation}
    P(\mathbb P,\mathbb Q)
    \coloneqq
    \sup_{\rho}
    \left[
        1-
        F_{\mathrm{BC}}(\mathbb P(\rho),\mathbb Q(\rho))^2
    \right]^{1/2},
\end{equation}
where $F_{\mathrm{BC}}$ is the Bhattacharyya coefficient between classical probability measures~\cite{lecamContiguityHellingerTransforms1990,tsybakovLowerBoundsMinimax2009}.
An apparatus state $\eta^A$ and a measurement $\mathbb M^{SA\sto\Omega}$ implement $\mathbb P^{S\sto\Omega}$ with error $\epsilon$ if
\begin{equation}
    P\bigl(\mathbb P^{S\sto\Omega},
    \mathbb M^{SA\sto\Omega}(\bullet\otimes\eta^A)\bigr)
    \le \epsilon .
    \label{eq:implementation_measurement}
\end{equation}

Similarly, an apparatus state $\eta^A$ and a channel $\Lambda^{SA\sto S'}$ implement a channel $\mathcal E^{S\sto S'}$ with error $\epsilon$ if
\begin{equation}
    P_e\bigl(\mathcal E^{S\sto S'},
    \Lambda^{SA\sto S'}(\bullet\otimes\eta^A)\bigr)
    \le \epsilon ,
    \label{eq:implementation_channel}
\end{equation}
where $P_e$ is the channel purified distance~\cite{gilchristDistanceMeasuresCompare2005},
\begin{equation}
    P_e(\mathcal E,\mathcal F)
    \coloneqq
    \sup_{R,\rho}
    P\!\left(
        (\mathcal E\otimes\id^R)(\rho),
        (\mathcal F\otimes\id^R)(\rho)
    \right),
\end{equation}
and $P(\rho,\sigma)=\sqrt{1-F(\rho,\sigma)^2}$ is the purified distance.

Let $G$ be a symmetry group.
For each $g\in G$, let $V_g^S,V_g^{S'},V_g^A$ be unitary or antiunitary projective representations on $S,S'$, and $A$, respectively; for each fixed $g$, the three operators are assumed to be either all unitary or all antiunitary.
We write
\[
    \mathcal V_g^X(\bullet)=V_g^X\bullet (V_g^X)^{-1}
    \qquad (X=S,S',A).
\]
The symmetry action on the outcome space is described by (measurable) bijections $T_g^\Omega:\Omega\to\Omega$, inducing the pushforward action
\[
    \mathcal T_g^\Omega:\mu\mapsto \mu\circ (T_g^\Omega)^{-1}.
\]
A measurement $\mathbb M^{SA\sto\Omega}$ and a channel $\Lambda^{SA\sto S'}$ are symmetric, or covariant~\cite{holevoCovariantMeasurementsOptimality2011} if
\begin{align}
    \mathbb M^{SA\sto\Omega}\circ
    (\mathcal V_g^S\otimes\mathcal V_g^A)
    &=
    \mathcal T_g^\Omega\circ
    \mathbb M^{SA\sto\Omega},
    \label{eq:covariant_measurement}
    \\
    \Lambda^{SA\sto S'}\circ
    (\mathcal V_g^S\otimes\mathcal V_g^A)
    &=
    \mathcal V_g^{S'}\circ
    \Lambda^{SA\sto S'}
    \label{eq:covariant_channel}
\end{align}
for all $g\in G$.
In the terminology of the resource theory of asymmetry~\cite{gourResourceTheoryQuantum2008,
marvianmashhadSymmetryAsymmetryQuantum2012,
marvianCoherenceDistillationMachines2020,
marvianOperationalInterpretationQuantum2022,
yamaguchiIidResourceTheory2023,
shitaraIidStateConvertibility2025,
yamaguchiQuantumGeometricTensor2024}, which is a branch of resource theory treating symmetry breaking as a resource, the covariant processors
defined in \eqref{eq:covariant_measurement} and \eqref{eq:covariant_channel} are the free operations. The only object that may break the
symmetry is the apparatus state $\eta^A$, which therefore plays the role of an
asymmetry resource, or equivalently a quantum reference frame. We ask how accurately a noncovariant target measurement or unitary gate can be implemented by such a free processor when assisted by this resource state $\eta^A$ in the apparatus.

\paragraph{Quantitative WAY theorem for general symmetries}
Our main results show that accurate implementation of a symmetry-breaking PVM or unitary gate requires the apparatus state to be distinguishable from its symmetry transform.

\begin{thm}
\label{thm:main_PVM}
Let $\mathbb P=\mathbb P^{S\sto\Omega}$ be a PVM.
Suppose that an apparatus state $\eta^A$ and a covariant measurement $\mathbb M^{SA\sto\Omega}$ implement $\mathbb P$ with error $\epsilon$.
Then, for every $g\in G$ such that the denominator below is nonzero,
\begin{equation}
    F(\eta^A,\mathcal V_g^A(\eta^A))
    \le
    \frac{2\epsilon}{
    \sup_E
    \norm{
        \mathbb P_E
        -
        \mathcal V_g^S
        \bigl(
            \mathbb P_{(T_g^\Omega)^{-1}(E)}
        \bigr)
    }}
    .
    \label{eq:main_PVM}
\end{equation}
\end{thm}

\begin{thm}
\label{thm:main_unitary}
Let $\mathcal U=\mathcal U^{S\sto S'}$ be a unitary gate.
Suppose that an apparatus state $\eta^A$ and a covariant channel $\Lambda^{SA\sto S'}$ implement $\mathcal U$ with error $\epsilon$.
Then, for every $g\in G$ such that the denominator below is nonzero,
\begin{equation}
    F(\eta^A,\mathcal V_g^A(\eta^A))
    \le
    \frac{4\epsilon}{
    \norm{
        \mathcal U
        -
        \mathcal V_g^{S'}
        \circ
        \mathcal U
        \circ
        \mathcal V_{g^{-1}}^S
    }_\diamond}
    .
    \label{eq:main_unitary}
\end{equation}
\end{thm}

These theorems quantify the amount of asymmetry in the ancilla state $\eta^A$ that is necessary for implementing non-covariant measurements or channels under general symmetry constraints. 
In particular, in the case of exact implementation, 
they recover the qualitative statement that a perfectly asymmetric state satisfying
$F(\eta^A,\mathcal{V}_g^A(\eta^A))=0$ is necessary~\cite{marvianInformationtheoreticAccountWignerArakiYanase2012}. 
In general, the fidelity $F(\eta^A,\mathcal V_g^A(\eta^A))$ quantifies how well the apparatus state can serve as a reference for the symmetry element $g$.
A small fidelity means that the apparatus state is nearly distinguishable from its transformed copy and hence carries strong asymmetry.
Indeed, when ${\cal V}$ is a projective unitary representation, 
its logarithmic quantity $L_g(\eta^A)
    \coloneqq
    -\ln F(\eta^A,\mathcal{V}_g^A(\eta^A))$
is an additive monotone that completely characterizes i.i.d.\ state convertibility for pure states~\cite{shitaraIidStateConvertibility2025,kusukiResourceTheoreticQuantifiersWeak2026}. 
We note that this quantity can take the value $+\infty$ and is therefore not continuous in general; consequently, 
it lies outside the scope of the general resource-theoretic WAY theorem of Ref.~\cite{tajimaUniversalTradeoffRelations2025}.

\paragraph{No-programming mechanism}
We now explain the mechanism behind Theorems~\ref{thm:main_PVM} and~\ref{thm:main_unitary}.
If Eq.~\eqref{eq:implementation_measurement} holds and the implementing measurement is covariant, then Eq.~\eqref{eq:covariant_measurement} implies
\begin{align}
P\!\left(
    \mathcal T_g^\Omega\circ
    \mathbb P^{S\sto\Omega}
    \circ
    \mathcal V_{g^{-1}}^S,
    \mathbb M^{SA\sto\Omega}
    (\bullet\otimes\mathcal V_g^A(\eta^A))
\right)
\le \epsilon .
\label{eq:twisted_implementation_measurement}
\end{align}
Similarly, if Eq.~\eqref{eq:implementation_channel} holds and the implementing channel is covariant, then
\begin{align}
P_e\!\left(
    \mathcal V_g^{S'}
    \circ
    \mathcal E^{S\sto S'}
    \circ
    \mathcal V_{g^{-1}}^S,
    \Lambda^{SA\sto S'}
    (\bullet\otimes\mathcal V_g^A(\eta^A))
\right)
\le \epsilon .
\label{eq:twisted_implementation_channel}
\end{align}
Thus the same processor approximately implements two targets using the two program states
\(
    \eta^A,
    \mathcal V_g^A(\eta^A).
\)
The WAY problem is therefore reduced to a two-target no-programming problem (Fig.~\ref{fig:2}).
\begin{figure}[tbp]
    \centering
    \includegraphics[width=\linewidth]{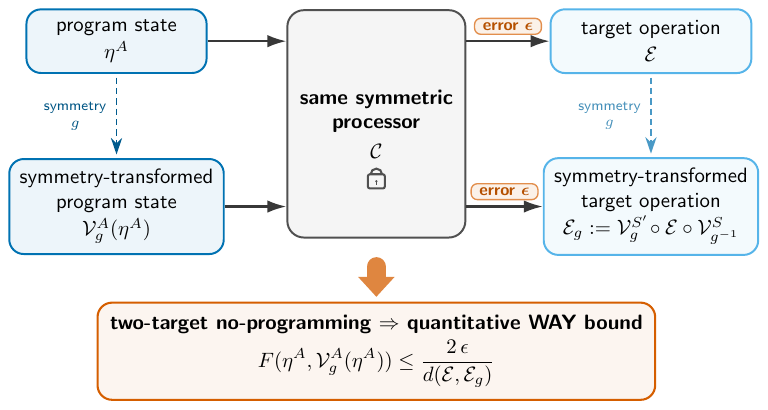}
    \caption{
No-programming origin of the quantitative WAY bound.
By covariance, the same symmetric processor that implements $\mathcal E$
from $\eta^A$ implements
$\mathcal E_g := \mathcal V_g^{S'}\circ\mathcal E\circ\mathcal V_{g^{-1}}^S$
from $\mathcal V_g^A(\eta^A)$.
The two apparatus states therefore serve as programs for two target operations.
A two-target no-programming inequality bounds
$F(\eta^A,\mathcal V_g^A(\eta^A))$ in terms of the implementation error
$\epsilon$ and the distance between $\mathcal E$ and $\mathcal E_g$.
Thus a small error is possible only if the apparatus carries sufficient
asymmetry.}
    \label{fig:2}
\end{figure}

For two measurements $(\mathbb P,\mathbb Q)$, 
let $d(\mathbb P,\mathbb Q)$ quantify their ability to amplify distinguishability: 
it is the largest input fidelity for which the two measurements can produce perfectly distinguishable outcome distributions. For two channels $(\mathcal E,\mathcal F)$, 
$d(\mathcal E,\mathcal F)$ is defined analogously, 
with optimization over reference systems.
Thus $d>0$ means that the two target operations can amplify distinguishability: they can take nonorthogonal inputs to perfectly distinguishable outputs.
For PVMs and unitary gates, we have
\begin{align}
d(\mathbb P,\mathbb Q)&=\sup_E\norm{\mathbb P_E-\mathbb Q_E},\label{eq:dist_PVM}\\
d(\mathcal U,\mathcal V)&=\frac12\norm{\mathcal U-\mathcal V}_\diamond .\label{eq:dist_unitary}
\end{align}
The precise definitions of $d$, as well as the proof of Eqs.~\eqref{eq:dist_PVM} and~\eqref{eq:dist_unitary}, 
are given in the End Matter.

The technical core is the following two-target no-programming bound.
Suppose that the same processor implements two targets $X_1$ and $X_2$ with program states $\eta_1^A$ and $\eta_2^A$, respectively, 
each within error $\epsilon$, 
where $X_i$ denotes either a measurement or a channel.
Then, whenever $d(X_1,X_2)>0$, we have
\begin{equation}
F(\eta_1^A,\eta_2^A)
\le
\frac{2\epsilon}{d(X_1,X_2)} .
\label{eq:NP_general}
\end{equation}
The proof is given in the End Matter.
Using this bound with
$\eta_2^A=\mathcal V_g^A(\eta_1^A)$
and choosing $X_2$ as the corresponding symmetry-transformed target, 
we obtain Theorems~\ref{thm:main_PVM} and~\ref{thm:main_unitary} by applying
Eqs.~\eqref{eq:dist_PVM} and~\eqref{eq:dist_unitary}.
Note that, for unitary gates, this no-programming bound strengthens previously known approximate-programming bounds~\cite{hilleryApproximateProgrammableQuantum2006,majenzEntropyQuantumInformation2017} in certain regimes; 
see Supplemental Material~\cite{SM}.

\paragraph{Discrete symmetry: position measurement on a lattice ring}
As a simple illustration of the theorem, consider the measurement of the position of a particle $S$ on a one-dimensional lattice ring with $L$ sites. 
To measure its position, introduce another particle $A$ on the same ring and use it as a probe.
Assume that the interaction between $S$ and $A$ and the final readout depend only on their relative position.
Equivalently, the implementation is covariant under lattice translations generated by $\tau$ on $S$ and $A$, while the outcome space $\Omega=\{1,\ldots,L\}$ is acted on trivially.

Let $(P_i)_{i=1}^L$ be the position PVM of $S$.
Then, for every nontrivial lattice translation,
\begin{equation}
    \norm{P_i-[(\tau^S)^\dagger]^n(P_i)}=
    \norm{P_i-P_{i+n}}=1\ (1\le n< L),
\end{equation}
where the site labels are understood modulo $L$. 
Theorem~\ref{thm:main_PVM} therefore implies that any implementation with error $\epsilon$ satisfies
\begin{equation}
    F(\eta^A,(\tau^A)^n(\eta^A))\le 2\epsilon\ (1\le n< L).
\end{equation}
Thus, high accuracy requires the probe state $\eta^A$ to be almost perfectly distinguishable from all of its nontrivial lattice translations. 
In other words, the ancilla must itself provide a sharply localized reference for the absolute position on the ring. 
This model can be regarded as a discrete analogue of the limitation on position measurements under momentum conservation~\cite{ozawaDoesConservationLaw1991,
buschPositionMeasurementsObeying2011,
kuramochiWignerArakiYanaseTheoremContinuous2023},
and shows that a tradeoff between the localization of the ancilla and the measurement accuracy persists even in the absence of a continuous symmetry.

\paragraph{Antiunitary symmetry: imaginarity and universality transformation}
The theorem also applies to antiunitary symmetries.
Consider complex-conjugation symmetry in a fixed computational basis.
Operations covariant under this symmetry are real operations, and the corresponding resource is called imaginarity~\cite{wuResourceTheoryImaginarity2021}.
This setting appears in universality transformations, where one asks whether a gate set consisting only of real orthogonal gates, such as one generated by $\{H,CCZ\}$, can become universal when supplemented with a complex ancillary state~\cite{takeuchiCatalyticTransformationComputationally2024,nakayamaUniquenessImaginarityassistedTransformation2026}.

Let $\mathcal U=\mathcal U^{S\to S'}$ be a target unitary gate and suppose that it is implemented with error $\epsilon$ by a complex-conjugation covariant channel $\Lambda^{SA\to S'}$ using an ancillary state $\eta^A$.
Then Theorem~\ref{thm:main_unitary} gives
\begin{equation}
    F(\eta^A,(\eta^A)^*)
    \le
    \frac{4\epsilon}{
        \norm{\mathcal U-\mathcal U^*}_\diamond
    },
    \label{eq:complex_conjugation_bound}
\end{equation}
where $\mathcal U^*$ denotes the complex conjugate of $\mathcal U$ in the computational basis. 
Thus, implementing a unitary that is far from real requires an ancillary state that is nearly distinguishable from its complex conjugate, 
which quantifies the no-go result in Ref.~\cite{nakayamaUniquenessImaginarityassistedTransformation2026}. 

For pure ancillary states, the above bound is tight in its scaling.
More precisely, for any fixed target unitary satisfying 
\(\mathcal U \ne \mathcal U^*\), the minimum achievable error scales as $\Theta\left(F(\eta^A,(\eta^A)^*)\right)$. 
To see this, note that~\cite{wuResourceTheoryImaginarity2021} complex-conjugation covariant operations can transform
$\eta^A$ into the maximally imaginary state
$\ket{i}\coloneqq\frac{\ket{0}+i\ket{1}}{\sqrt 2}$
with error
$
\sqrt{\frac{1}{2}
\left[
    1-P\!\left(\eta^A,(\eta^A)^*\right)
\right]}
\in
O\!\left(F\!\left(\eta^A,(\eta^A)^*\right)\right)$.
Since \(\ket{i}\) is sufficient to implement any unitary gate~\cite{takeuchiCatalyticTransformationComputationally2024,nakayamaUniquenessImaginarityassistedTransformation2026}, this gives
an achievable protocol whose error has the same scaling as the lower bound.

\paragraph{Applications to symmetries generated by unbounded observables}
Although the preceding motivation focuses on going beyond continuous generators,
our framework also covers continuous symmetries, including those generated by
unbounded operators. In particular, Theorems~\ref{thm:main_PVM}
and~\ref{thm:main_unitary} apply without any domain assumptions.
In the exact case $\epsilon=0$, they therefore reproduce
the recently established qualitative WAY theorem
for continuous symmetries including unbounded conserved quantities~\cite{kuramochiWignerArakiYanaseTheoremContinuous2023}.
We note that the denominator on the right-hand side of Eq.~\eqref{eq:main_PVM} is positive along some sequence $\{t_n\}$ satisfying $t_n\to0$, 
provided that $\mathbb{P}$ does not commute with $V_t^S=e^{it\hat{L}^S}$ for some $t\in\mathbb{R}$. 
Otherwise, there would exist some $t_0>0$ such that $[\mathbb{P}_E,V_t^S]=0$ for all $t\in(-t_0,t_0)$ and all $E$. 
This would imply that $\mathbb{P}_E$ commutes with $V_t^S$ for all $t\in\mathbb{R}$ and all $E$, contradicting the assumption.
The same argument applies to the unitary WAY case.
For general $\epsilon$, our theorems provide a quantitative refinement of 
Ref.~\cite{kuramochiWignerArakiYanaseTheoremContinuous2023}.

The unboundedness of the generator may in fact prohibit even
approximate implementation.
Consider, for instance, the position measurement of a one-dimensional
particle on the line under momentum covariance, more precisely,
$G=\R$, $V_s^{X}=e^{-is\hat{p}^X}$, and $T_s^\R=\id_\R$.
For any $s\neq0$,
\begin{equation}
    \norm{\mathbb{P}^{S\to \R}_{[0,\infty)}
    -e^{-is\hat{p}^S}
    \mathbb{P}^{S\to \R}_{[0,\infty)}
    e^{is\hat{p}^S}}
    =
    \norm{\mathbb{P}^{S\to \R}_{[0,\infty)}
    -\mathbb{P}^{S\to \R}_{[-s\hbar,\infty)}}
    =1 .
\end{equation}
On the other hand,
$F(\eta^A,\mathcal{V}_s^A(\eta^A))\to1$ as $s\to0$.
Hence Eq.~\eqref{eq:main_PVM} can hold only if
$\epsilon\ge 1/2$.

The same phenomenon occurs for the projective measurement of a
quadrature amplitude of a single-mode optical field under conservation
of the total photon number; see Supplemental Material~\cite{SM}.
The corresponding qualitative obstruction was discussed in
Ref.~\cite{kuramochiWignerArakiYanaseTheoremContinuous2023}.

These examples show that no approximate implementation can be uniformly
accurate over all input states.
For the latter example, this is consistent with the fact that homodyne measurements, 
although known to be asymptotically exact for each fixed input state~\cite{yuenQuantumStatisticsHomodyne1978}, 
have an error that increases with the photon number of the 
measured state~\cite{braunsteinHomodyneStatistics1990,tycOperationalFormulationHomodyne2004}.
To obtain a meaningful error measure in such cases, one must use a
state-dependent measure, 
namely one defined by restricting the
input states to a subset of the state space, 
such as energy-constrained
states.
Our no-programming theorem and WAY theorem admit a formal extension to
this setting as well; 
see Supplemental Material~\cite{SM}.

\paragraph{Conclusion and outlook}
We have derived quantitative WAY-type bounds for symmetry-breaking projective measurements and unitary gates under arbitrary unitary and antiunitary symmetries.
The bounds show that accurate implementation requires the apparatus state to be distinguishable from its symmetry-transformed copy.
This replaces the generator variance used in conventional continuous-symmetry WAY bounds by an operational asymmetry measure that remains meaningful for discrete and antiunitary symmetries.

The proof rests on a two-target no-programming inequality: two distinguishability-amplifying target operations cannot be approximately implemented by the same processor unless the corresponding program states are sufficiently distinguishable.
Applied to covariant processors, this directly yields quantitative WAY tradeoffs.

Several questions remain open.
One is to construct protocols that achieve the bounds, or determine their optimal constants, for general symmetries.
Another is to develop state-dependent versions suited to unbounded generators and energy-constrained inputs.
It would also be interesting to understand whether the connection between WAY-type limitations and covariant quantum error correction, known for continuous symmetries~\cite{tajimaUniversalTradeoffStructure2022}, extends to discrete and antiunitary symmetries.

\paragraph{Acknowledgments}
The authors thank Yui Kuramochi for helpful comments on an earlier draft of the manuscript.
A.H. was supported by KAKENHI Grant
No. JP25KJ0833 from the Japan Society for the
Promotion of Science (JSPS) and FoPM, a WINGS Program, 
the University of Tokyo.
A.H. also acknowledges support from JSR Fellowship, the University of Tokyo.
H.T. was supported by JSPS Grants-in-Aid for Scientific Research No. JP25K00924, MEXT KAKENHI Grant-in-Aid for Transformative
Research Areas B ``Quantum Energy Innovation” Grant Numbers 24H00830 and 24H00831, JST FOREST No. JPMJFR2365, JST MOONSHOT No. JPMJMS256E and Royal Society International Collaboration Awards 2025 Flexigrant number ICA/R2/252240.

\bibliographystyle{apsrev4-2-titles-revised}
\bibliography{WAY,suppl}

\clearpage
\section*{END MATTER}
\subsection{Precise definitions of $d(\mathbb P,\mathbb Q),d(\mathcal E,\mathcal F)$}
For a pair of measurements $(\mathbb P,\mathbb Q)$ and a pair of channels $(\mathcal E,\mathcal F)$, define
\begin{align}
 \hat F(\lambda;\mathbb P,\mathbb Q)
 &\coloneqq
 \inf_{\substack{\rho,\sigma\\ F(\rho,\sigma)\ge \lambda}}
 F_{\mathrm{BC}}\!\left(\mathbb P(\rho),\mathbb Q(\sigma)\right),
 \label{eq:Fhat_measurement}
 \\
 \hat F_e(\lambda;\mathcal E,\mathcal F)
 &\coloneqq
 \inf_{\substack{R,\rho,\sigma\\ F(\rho,\sigma)\ge \lambda}}
 F\!\left(
 (\mathcal E\otimes\id^R)(\rho),
 (\mathcal F\otimes\id^R)(\sigma)
 \right),
 \label{eq:Fhat_channel}
\end{align}
and
\begin{align}
 d(\mathbb P,\mathbb Q)
 &\coloneqq
 \sup\!\left\{
 \lambda\in[0,1]\ \middle|\
 \hat F(\lambda;\mathbb P,\mathbb Q)=0
 \right\},
 \\
 d(\mathcal E,\mathcal F)
 &\coloneqq
 \sup\!\left\{
 \lambda\in[0,1]\ \middle|\
 \hat F_e(\lambda;\mathcal E,\mathcal F)=0
 \right\}.
 \label{eq:orthogonality_amplification}
\end{align}

\subsection{Proof of Eq.~\eqref{eq:dist_PVM}}

We first recall the following elementary identity.
For any two projections $P$ and $Q$,
\begin{equation}
    \norm{P-Q}
    =
    \max\{\norm{P(I-Q)},\norm{(I-P)Q}\}.
\end{equation}
Indeed,
\begin{equation}
    (P-Q)^2
    =
    P(I-Q)^2P+(I-P)Q^2(I-P),
\end{equation}
and the two terms on the right-hand side have orthogonal supports. Hence
\begin{align}
    \norm{P-Q}^2
    &=
    \norm{P(I-Q)^2P+(I-P)Q^2(I-P)}
    \nonumber\\
    &=
    \max\{\norm{P(I-Q)^2P},\norm{(I-P)Q^2(I-P)}\}
    \nonumber\\
    &=
    \max\{\norm{P(I-Q)},\norm{(I-P)Q}\}^2 .
\end{align}
Applying this identity to $\mathbb{P}_E$ and $\mathbb{Q}_E$, 
we obtain
\begin{equation}
    \sup_E \norm{\mathbb{P}_E-\mathbb{Q}_E}
    =
    \sup_E \norm{\mathbb{P}_E\mathbb{Q}_{\Omega\setminus E}} .
\end{equation}

We next prove that the right-hand side is equal to
$d(\mathbb{P},\mathbb{Q})$.
For each (measurable) subset $E\subset\Omega$,
\begin{align}
    \norm{\mathbb{P}_E\mathbb{Q}_{\Omega\setminus E}}
    &=
    \sup_{\|\psi\|=\|\phi\|=1}
    \abs{\bra{\psi}\mathbb{P}_E\mathbb{Q}_{\Omega\setminus E}\ket{\phi}}
    \nonumber\\
    &=
    \sup_{\substack{\|\psi\|=\|\phi\|=1\\
    \mathbb{P}_E\ket{\psi}=\ket{\psi},\;
    \mathbb{Q}_{\Omega\setminus E}\ket{\phi}=\ket{\phi}}}
    \abs{\braket{\psi}{\phi}},
\end{align}
where the second supremum is understood to be zero if the constraint set is
empty. If
$\mathbb{P}_E\ket{\psi}=\ket{\psi}$ and
$\mathbb{Q}_{\Omega\setminus E}\ket{\phi}=\ket{\phi}$, then
\begin{equation}
    F_{\mathrm{BC}}\bigl(\mathbb{P}(\psi),\mathbb{Q}(\phi)\bigr)=0 .
\end{equation}
Therefore,
\begin{equation}
    \hat{F}\bigl(\abs{\braket{\psi}{\phi}};\mathbb{P},\mathbb{Q}\bigr)=0,
\end{equation}
and hence
\begin{equation}
    \sup_E \norm{\mathbb{P}_E\mathbb{Q}_{\Omega\setminus E}}
    \le
    d(\mathbb{P},\mathbb{Q}) .
\end{equation}

Conversely, suppose that
$\hat{F}(\lambda;\mathbb{P},\mathbb{Q})=0$.
Then, for each $n\ge 1$, there exist states $\rho_n$ and $\sigma_n$ such that
\begin{equation}
    F(\rho_n,\sigma_n)\ge \lambda,
    \qquad
    F_{\mathrm{BC}}\bigl(\mathbb{P}(\rho_n),\mathbb{Q}(\sigma_n)\bigr)
    <
    \frac{1}{4n^2}.
\end{equation}
By Le Cam's inequality~\cite{lecamContiguityHellingerTransforms1990,tsybakovLowerBoundsMinimax2009},
\begin{equation}
    1-F_{\mathrm{BC}}(\mu,\nu)
    \le
    \sup_E \{\mu(E)-\nu(E)\},
\end{equation}
there exists a measurable set $E_n$ such that
\begin{align}
    \Tr(\mathbb{P}_{E_n}\rho_n)
    &>
    1-\frac{1}{4n^2},\\
    \Tr(\mathbb{Q}_{E_n}\sigma_n)
    &<
    \frac{1}{4n^2}.
\end{align}
By Uhlmann's theorem~\cite{uhlmannTransitionProbabilityState1976}, we can choose purifications
$\ket{\Psi_n}$ and $\ket{\Phi_n}$ of $\rho_n$ and $\sigma_n$, respectively,
such that
\begin{equation}
    \abs{\braket{\Psi_n}{\Phi_n}}
    =
    F(\rho_n,\sigma_n).
\end{equation}
Thus,
\begin{align}
    \lambda&\le \abs{\braket{\Psi_n}{\Phi_n}}\nonumber\\
    &\le \abs{\braket[3]{\Psi_n}{\mathbb{P}_{E_n}\mathbb{Q}_{\Omega\setminus E_n}\otimes I}{\Phi_n}}+\frac{1}{n}\nonumber\\
    &\le \sup_E\norm{\mathbb{P}_E\mathbb{Q}_{\Omega\setminus E}}+\frac{1}{n}.
\end{align}
Taking $n\to\infty$ gives
\begin{equation}
    \lambda
    \le
    \sup_E \norm{\mathbb{P}_E\mathbb{Q}_{\Omega\setminus E}} .
\end{equation}
Finally, taking the supremum over all
$\lambda$ satisfying
$\hat{F}(\lambda;\mathbb{P},\mathbb{Q})=0$, we obtain
\begin{equation}
    d(\mathbb{P},\mathbb{Q})
    \le
    \sup_E \norm{\mathbb{P}_E\mathbb{Q}_{\Omega\setminus E}} .
\end{equation}
Combining the two inequalities proves Eq.~\eqref{eq:dist_PVM}.

\subsection{Proof of Eq.~\eqref{eq:dist_unitary}}

We first note that, by purification and Uhlmann's theorem~\cite{uhlmannTransitionProbabilityState1976},
\begin{equation}
    \hat{F}_e(\lambda;\mathcal{E},\mathcal{F})
    =
    \inf_R
    \inf_{\substack{\psi,\phi:\ {\rm pure}\\
    \abs{\braket{\psi}{\phi}}\ge \lambda}}
    F\bigl(
        \mathcal{E}\otimes\id^R(\psi),
        \mathcal{F}\otimes\id^R(\phi)
    \bigr).
\end{equation}
In particular,
\begin{align}
    \hat{F}_e(\lambda;\id,\mathcal{F})
    &=
    \inf_R
    \inf_{\substack{\psi,\phi:\ {\rm pure}\\
    \abs{\braket{\psi}{\phi}}\ge \lambda}}
    F\bigl(
        \psi,
        \mathcal{F}\otimes\id^R(\phi)
    \bigr)
    \nonumber\\
    &=
    \left[
    \inf_R
    \inf_{\substack{\psi,\phi:\ {\rm pure}\\
    \abs{\braket{\psi}{\phi}}\ge \lambda}}
    \bra{\psi}
        \mathcal{F}\otimes\id^R(\phi)
    \ket{\psi}
    \right]^{1/2}.
    \label{eq:fidelity_unitary}
\end{align}

We now prove Eq.~\eqref{eq:dist_unitary}. Since
\begin{equation}
    d(\mathcal{U},\mathcal{V})
    =
    d(\id,\mathcal{U}^{-1}\circ\mathcal{V}),
\end{equation}
and since the diamond norm is invariant under unitary pre- and post-processing,
it suffices to consider the case $\mathcal{U}=\id$.
Let $\mathcal{V}(\cdot)=V(\cdot)V^\dagger$. Then
Eq.~\eqref{eq:fidelity_unitary} gives
\begin{equation}
    \hat{F}_e(\lambda;\id,\mathcal{V})
    =
    \inf_R
    \inf_{\substack{\psi,\phi:\ {\rm pure}\\
    \abs{\braket{\psi}{\phi}}\ge \lambda}}
    \abs{
        \bra{\psi}
        V\otimes I^R
        \ket{\phi}
    } .
\end{equation}

We next evaluate the inner infimum. For a unitary $W$, define
\begin{equation}
    \mu(W)
    \coloneqq
    \inf_{\phi:\ {\rm pure}}
    \abs{\bra{\phi}W\ket{\phi}} .
\end{equation}
Then
\begin{align}
    &\inf_{\substack{\psi,\phi:\ {\rm pure}\\
    \abs{\braket{\psi}{\phi}}\ge \lambda}}
    \abs{\bra{\psi}W\ket{\phi}}
    \nonumber\\
    &=
    \inf_{\substack{\phi,\eta:\ {\rm pure}\\
    \braket{\phi}{\eta}=0}}
    \inf_{t\ge \lambda}
    \abs{
        \bigl(t\bra{\phi}+\sqrt{1-t^2}\bra{\eta}\bigr)
        W
        \ket{\phi}
    }
    \nonumber\\
    &=
    \inf_{\phi:\ {\rm pure}}
    \inf_{t\ge \lambda}
    \left[
        t\abs{\bra{\phi}W\ket{\phi}}
        -
        \sqrt{1-t^2}
        \sqrt{1-\abs{\bra{\phi}W\ket{\phi}}^2}
    \right]_+
    \nonumber\\
    &=
    \left[
        \lambda\mu(W)
        -
        \sqrt{1-\lambda^2}
        \sqrt{1-\mu(W)^2}
    \right]_+ ,
\end{align}
where $[x]_+\coloneqq \max\{x,0\}$.
Applying this identity to $W=V\otimes I^R$, we obtain
\begin{align}
    &\hat{F}_e(\lambda;\mathcal{U},\mathcal{V})=0\nonumber\\
    &\Leftrightarrow\lambda\inf_R\mu(V\otimes I^R)\le\sqrt{1-\lambda^2}\sqrt{1-\inf_R\mu(V\otimes I^R)^2}\nonumber\\
    &\Leftrightarrow
    \lambda\le \sqrt{1-\inf_R\mu(V\otimes I^R)^2}.
\end{align}
Therefore,
\begin{equation}
    d(\id,\mathcal{V})
    =
    \sup_R
    \sqrt{1-\mu(V\otimes I^R)^2}.
\end{equation}

On the other hand,
\begin{align}
    \sup_R
    \sqrt{1-\mu(V\otimes I^R)^2}
    &=
    \sup_R
    \sup_{\phi:\ {\rm pure}}
    \sqrt{
        1-
        \abs{\bra{\phi}V\otimes I^R\ket{\phi}}^2
    }
    \nonumber\\
    &=
    \frac{1}{2}
    \sup_R
    \sup_{\phi:\ {\rm pure}}
    \norm{
        \phi-
        \mathcal{V}\otimes\id^R(\phi)
    }_1
    \nonumber\\
    &=
    \frac{1}{2}
    \norm{\id-\mathcal{V}}_\diamond .
\end{align}
Combining this with the reduction to $\mathcal{U}=\id$ proves
Eq.~\eqref{eq:dist_unitary}.

\subsection{Proof of Eq.~\eqref{eq:NP_general}}
We show the following stronger statements. 
\begin{thm}
\label{thm:NP_measurement}
Suppose that, for each $i=1,2$, the pair
$(\mathbb{M}^{SA\sto\Omega},\eta^A_i)$ implements the measurement
$\mathbb{P}^{S\sto\Omega}_i$ with error at most $\epsilon$.
Then
\begin{equation}
 F(\eta^A_1,\eta^A_2)
 \le
 \inf_{\lambda>0}
 \frac{
 \hat{F}(\lambda;\mathbb{P}^{S\sto\Omega}_1,
 \mathbb{P}^{S\sto\Omega}_2)+2\epsilon
 }{\lambda}.
 \label{eq:pre_NP_measurement}
\end{equation}
In particular, if
$d(\mathbb{P}^{S\sto\Omega}_1,\mathbb{P}^{S\sto\Omega}_2)>0$, then
we have
\begin{equation}
 F(\eta^A_1,\eta^A_2)
 \le
 \frac{2\epsilon}{
 d(\mathbb{P}^{S\sto\Omega}_1,\mathbb{P}^{S\sto\Omega}_2)
 } .
 \label{eq:NP_measurement_EM}
\end{equation}
\end{thm}

\begin{thm}
\label{thm:NP_channel}
Suppose that, for each $i=1,2$, the pair
$(\Lambda^{SA\sto S'},\eta^A_i)$ implements the channel
$\mathcal{E}^{S\sto S'}_i$ with error at most $\epsilon$.
Then
\begin{equation}
 F(\eta^A_1,\eta^A_2)
 \le
 \inf_{\lambda>0}
 \frac{
 \hat{F}_e(\lambda;\mathcal{E}^{S\sto S'}_1,
 \mathcal{E}^{S\sto S'}_2)+2\epsilon
 }{\lambda}.
 \label{eq:pre_NP_channel}
\end{equation}
In particular, if
$d(\mathcal{E}^{S\sto S'}_1,\mathcal{E}^{S\sto S'}_2)>0$, then
we have 
\begin{equation}
 F(\eta^A_1,\eta^A_2)
 \le
 \frac{2\epsilon}{
 d(\mathcal{E}^{S\sto S'}_1,\mathcal{E}^{S\sto S'}_2)
 } .
 \label{eq:NP_channel_EM}
\end{equation}
\end{thm}

\begin{proof}[Proof of Theorems~\ref{thm:NP_measurement}
and~\ref{thm:NP_channel}]
We first prove Theorem~\ref{thm:NP_measurement}.
For any states $\rho$ and $\sigma$, monotonicity of fidelity gives
\begin{align}
 &F(\rho,\sigma)F(\eta^A_1,\eta^A_2)
 =
 F(\rho\otimes\eta^A_1,\sigma\otimes\eta^A_2)
 \nonumber\\
 &\le
 F_{\mathrm{BC}}\!\left(
 \mathbb{M}^{SA\sto\Omega}(\rho\otimes\eta^A_1),
 \mathbb{M}^{SA\sto\Omega}(\sigma\otimes\eta^A_2)
 \right).
\end{align}
Using the continuity bound
\begin{equation}
 \left|
 F_{\mathrm{BC}}(\mu,\nu)
 -
 F_{\mathrm{BC}}(\mu,\nu')
 \right|
 \le
 P(\nu,\nu')
\end{equation}
in both arguments, together with the assumed implementation errors, we obtain
\begin{align}
 F(\rho,\sigma)F(\eta^A_1,\eta^A_2)
 &\le
 F_{\mathrm{BC}}\left(
 \mathbb{P}^{S\sto\Omega}_1(\rho),
 \mathbb{P}^{S\sto\Omega}_2(\sigma)
 \right)
 +2\epsilon .
\end{align}
Dividing by $F(\rho,\sigma)$ and optimizing over all
$\rho,\sigma$ satisfying $F(\rho,\sigma)\ge\lambda$ yields
Eq.~\eqref{eq:pre_NP_measurement}.
If
$d(\mathbb{P}^{S\sto\Omega}_1,\mathbb{P}^{S\sto\Omega}_2)>0$,
then, for every
$\lambda<d(\mathbb{P}^{S\sto\Omega}_1,\mathbb{P}^{S\sto\Omega}_2)$
with
$\hat{F}(\lambda;\mathbb{P}^{S\sto\Omega}_1,
\mathbb{P}^{S\sto\Omega}_2)=0$,
Eq.~\eqref{eq:pre_NP_measurement} gives
$F(\eta^A_1,\eta^A_2)\le 2\epsilon/\lambda$.
Taking $\lambda$ upward to the supremum gives
Eq.~\eqref{eq:NP_measurement_EM}.

The proof of Theorem~\ref{thm:NP_channel} is identical, with
$F$ and $\hat{F}_e$ replacing $F_{\mathrm{BC}}$ and $\hat{F}$, and with the
optimization also taken over an arbitrary reference system $R$.
\end{proof}

\clearpage\clearpage
\makeatletter
   	\c@secnumdepth=4
\makeatother

\setcounter{equation}{0}
\setcounter{figure}{0}
\setcounter{section}{0}
\setcounter{table}{0}
\setcounter{thm}{0} 
\renewcommand{\theequation}{S\arabic{equation}}
\renewcommand{\thefigure}{S\arabic{figure}}
\renewcommand{\theHequation}{\theequation}
\renewcommand{\theHfigure}{\thefigure}
\renewcommand{\thethm}{S\arabic{thm}}
\renewcommand{\thesection}{S\arabic{section}}

\renewcommand{\bibnumfmt}[1]{[S#1]}
\renewcommand{\citenumfont}[1]{S#1}

\renewcommand{\thepage}{S\arabic{page}}
\setcounter{page}{1}

\makeatletter
\let\orig@footnote\footnote
\long\def\SM@pagefootnote#1{%
  \refstepcounter{footnote}%
  \begingroup
    \protected@xdef\@thefnmark{\thefootnote}%
  \endgroup
  \@footnotemark
  \@footnotetext{#1}%
}
\let\footnote\SM@pagefootnote
\makeatother

\title{
  Supplemental Material for \protect\\
  ``\nameoftitle''
}

\maketitle
\onecolumngrid

This Supplemental Material complements the main text in three respects.
First, we compare the no-programming inequality for unitary channels used in
the main text with existing approximate no-programming bounds~\citeSM{SM_hilleryApproximateProgrammableQuantum2006,SM_majenzEntropyQuantumInformation2017}, clarifying the relations among the
different error measures. Second, we analyze the projective measurement of a
quadrature amplitude of a single-mode optical field under photon-number
conservation. As in the position-measurement example in the main text, this
measurement cannot be implemented, even approximately, with uniformly small
error. Third, we formulate a restricted-input version of the no-programming-type
bound and derive the corresponding WAY-type statement.

\section{Comparison of no-programming inequalities for unitary gates}

In this section, we compare our no-programming inequality for unitary
channels with two bounds in Refs.~\citeSM{SM_hilleryApproximateProgrammableQuantum2006} and~\citeSM{SM_majenzEntropyQuantumInformation2017}.

\subsection{Relation to the results of M.~Hillery \textit{et al.}}

In Ref.~\citeSM{SM_hilleryApproximateProgrammableQuantum2006},
the approximation error is measured by the process fidelity, namely the
squared fidelity between Choi states.  
Let
\begin{equation}
J(\mathcal E)
:=
(\mathcal E^{S\to S'}\otimes\id^{\bar S})(\Phi^{S\bar S})
\end{equation}
be the normalized state obtained through the Choi--Jamio\l{}kowski isomorphism, 
where $\Phi^{S\bar S}$ is a normalized maximally entangled state.  
We define
\begin{equation}
F^2_{\rm proc}(\mathcal E,\mathcal F)
:=
F(J(\mathcal E),J(\mathcal F))^2.
\end{equation}
The corresponding process purified distance is
\begin{equation}
p_{\rm proc}(\mathcal E,\mathcal F)
:=
P(J(\mathcal E),J(\mathcal F))
=
\sqrt{1-F^2_{\rm proc}(\mathcal E,\mathcal F)}.
\end{equation}

We now recall the bound in Ref.~\citeSM{SM_hilleryApproximateProgrammableQuantum2006} in our notation.

\begin{thm}[Eq.~(4.19) of~\citeSM{SM_hilleryApproximateProgrammableQuantum2006}]
Suppose that, for each $i=1,2$, a fixed processor $\Lambda^{SA\sto S'}$ and 
a pure apparatus state $\eta_i^A$ implement the unitary gate
$\mathcal{U}^{S\sto S'}_i$ with process error at most $\epsilon_{\rm proc}$:
\begin{equation}
p_{\rm proc}(\mathcal U_i^{S\sto S'},
    \Lambda^{SA\sto S'}(\bullet\otimes\eta_i^A))
\le
\epsilon_{\rm proc}
\qquad (i=1,2).
\end{equation}
Then
\begin{equation}
F(\eta_1^A,\eta_2^A)
\le
\frac{\abs{\Tr(U_1^\dag U_2)}}{d_S}
\cdot
\min\left(
1,
\frac{
2\left(2\sqrt{d_S}\epsilon_{\rm proc}+d_S\epsilon_{\rm proc}^2
\right)
}
{\norm{\mathcal{U}_1-\mathcal{U}_2}_\diamond}
\right)
+2\epsilon_{\rm proc}
+\epsilon_{\rm proc}^2,
\label{eq:HZB-bound-process}
\end{equation}
where $d_S$ is the dimension of the Hilbert space of $S$.\footnote{
The final displayed equation, Eq.~(4.19) of Ref.~\citeSM{SM_hilleryApproximateProgrammableQuantum2006}, contains the factor $1/\eta$, where
\begin{equation}
\eta=\max_\psi\left(1-\abs*{\bra{\psi}U_1^\dag U_2\ket{\psi}}^2\right)
=\norm{\mathcal{U}_1-\mathcal{U}_2}\diamond^2/4.
\end{equation}
This appears to be a typographical error, since the preceding estimate, Eq.~(4.17) of Ref.~\citeSM{SM_hilleryApproximateProgrammableQuantum2006}, yields
$\frac{1}{\sqrt{\eta}}
=\frac{2}{\norm{\mathcal{U}_1-\mathcal{U}_2}\diamond}$.
We therefore use the corrected factor $\frac{2}{\norm{\mathcal{U}_1-\mathcal{U}2}\diamond}$ in Eq.~\eqref{eq:HZB-bound-process}.
}
\end{thm}

If Eq.~\eqref{eq:HZB-bound-process} is expressed
in terms of our error $\epsilon$ using $\epsilon_{\mathrm{proc}}\le \epsilon$, 
Eq.~\eqref{eq:HZB-bound-process} yields a weaker bound than our inequality
\begin{equation}
    F(\eta_1^A,\eta_2^A)\le \frac{4\epsilon}{\norm{\mathcal{U}_1-\mathcal{U}_2}_\diamond}
    \label{eq:suppl_main}
\end{equation}
for small $\epsilon$.
Here we use 
$\abs{\Tr(U_1^\dag U_2)}/d_S\ge
\inf_\psi\abs*{\bra{\psi}U_1^\dag U_2\ket{\psi}}
=\sqrt{1-\norm{\mathcal{U}_1-\mathcal{U}_2}_\diamond^2/4}$.

Conversely, our inequality~\eqref{eq:suppl_main} gives the following process-fidelity
corollary:
\begin{equation}
F(\eta_1^A,\eta_2^A)
\le
\frac{4}{\norm{\mathcal{U}_1-\mathcal{U}_2}_\diamond}
\sqrt{
2d_S
(1-\sqrt{1-\epsilon_{\rm proc}^2})}
=\frac{4\sqrt{d_S}\epsilon_{\rm proc}}{\norm{\mathcal{U}_1-\mathcal{U}_2}_\diamond}
\bigl(1+O(\epsilon_{\rm proc}^2)\bigr).
\label{eq:our-bound-process-corollary}
\end{equation}
Here we use the following bound\footnote{
This bound follows from Uhlmann's theorem~\citeSM{SM_uhlmannTransitionProbabilityState1976} and the resulting convexity property.
Indeed, for any pure state $\Psi^{S\bar S}$, 
$F((\mathcal E^{S\to S'}\otimes\id^{\bar S})(\Psi^{S\bar S}),(\mathcal F^{S\to S'}\otimes\id^{\bar S})(\Psi^{S\bar S}))$
can be expressed,
by Uhlmann's theorem~\citeSM{SM_uhlmannTransitionProbabilityState1976},
as a convex function of $\Tr_{\bar S}\Psi^{S\bar S}$.
Then, Eq.~\eqref{eq:wc-vs-proc-unitary} follows from the fact that there exists a state $\sigma^S$ satisfying
$\frac{I}{d_S}=\frac{\Tr_{\bar S}\Psi^{S\bar S}+(d_S-1)\sigma^S}{d_S}$, 
where $\frac{I}{d_S}=\Tr_{\bar S}\Phi^{S\bar S}$ is the maximally mixed state.
}:
\begin{equation}
F_{\rm proc}(\mathcal E,\mathcal F)\le
\frac{\sqrt{1-P_e(\mathcal E,\mathcal F)^2}+d_S-1}{d_S}.
\label{eq:wc-vs-proc-unitary}
\end{equation}
For small $\epsilon_{\rm proc}$, Eq.~\eqref{eq:our-bound-process-corollary}
is tighter than Eq.~\eqref{eq:HZB-bound-process} if 
$\norm{\mathcal{U}_1-\mathcal{U}_2}_\diamond\le 
4\sqrt{d_S}/(d_S+1)
$.

\subsection{Relation to the results of C.~Majenz}

In Ref.~\citeSM{SM_majenzEntropyQuantumInformation2017},
the approximation error is measured by the diamond norm:
\begin{equation}
\norm{\mathcal U_i^{S\sto S'}-
    \Lambda^{SA\sto S'}(\bullet\otimes\eta_i^A)}_\diamond
\le
\epsilon_\diamond
\qquad (i=1,2).
\end{equation}
For a unitary target, we have
\begin{equation}
P_e(\mathcal E,\mathcal U)
\le
\sqrt{
\frac{
\norm{\mathcal E-\mathcal U}_\diamond
}{2}
}.
\label{eq:diamond-vs-wc}
\end{equation}
This follows from the Fuchs--van de Graaf inequality~\citeSM{SM_fuchsCryptographicDistinguishabilityMeasures1999}:
\begin{equation}
P(\rho,\sigma)^2\le \frac12\norm{\rho-\sigma}_1
\end{equation}
for pure $\sigma$.

\begin{thm}[Theorem~3.43 of~\citeSM{SM_majenzEntropyQuantumInformation2017}]
Suppose that, for each $i=1,2$, a fixed processor $\Lambda^{SA\sto S'}$ and 
a pure apparatus state $\eta_i^A$ implement the unitary gate
$\mathcal{U}^{S\sto S'}_i$ with diamond-norm error at most $\epsilon_{\diamond}$:
\begin{equation}
\norm{\mathcal U_i- \Lambda^{SA\sto S'}(\bullet\otimes\eta_i^A)}_\diamond
\le
\epsilon_\diamond
\qquad (i=1,2),
\end{equation}
and assume that $\norm{\mathcal{U}_1-\mathcal{U}_2}_\diamond<2$.
Then
\begin{equation}
F(\eta_1^A,\eta_2^A)
\le
\frac{
2\left(
\sqrt{2\epsilon_\diamond}
+7\epsilon_\diamond
+2\sqrt{2}\epsilon_\diamond^{3/2}
+3\epsilon_\diamond^2
\right)
}{\inf_{\theta\in\mathbb R}
\norm{U_1-e^{i\theta}U_2}_\infty}
+2\epsilon_\diamond+\epsilon_\diamond^2.
\label{eq:Majenz-bound}
\end{equation}
\end{thm}

Combining our inequality~\eqref{eq:suppl_main} with
Eq.~\eqref{eq:diamond-vs-wc} gives the diamond-norm corollary
\begin{equation}
F(\eta_1^A,\eta_2^A)
\le
\frac{2\sqrt{2\epsilon_\diamond}}{\norm{\mathcal{U}_1-\mathcal{U}_2}_\diamond}
\le 
\frac{2\sqrt{\epsilon_\diamond}}{\inf_{\theta\in\mathbb R}
\norm{U_1-e^{i\theta}U_2}_\infty}
.
\label{eq:our-diamond-corollary}
\end{equation}
For $0<\norm{\mathcal{U}_1-\mathcal{U}_2}_\diamond<2$,
Eq.~\eqref{eq:our-diamond-corollary}
is uniformly tighter than Eq.~\eqref{eq:Majenz-bound}.
Here we use the inequality~\citeSM{SM_majenzEntropyQuantumInformation2017}:
\begin{equation}
   \norm{\mathcal{U}_1-\mathcal{U}_2}_\diamond=
   2\gamma\sqrt{1-\gamma^2/4}
   \ge \sqrt{2}\gamma
   \quad (\gamma\coloneqq\inf_{\theta\in\mathbb R}
\norm{U_1-e^{i\theta}U_2}_\infty\in[0,\sqrt{2}])
\end{equation}
for $\norm{\mathcal{U}_1-\mathcal{U}_2}_\diamond<2$.

\section{Projective quadrature measurements under photon-number conservation}
Let $\hat{a}^S$ be the annihilation operator, 
and let
$\hat{q}^S
=
\frac{1}{2}
\left(\hat{a}^S+(\hat{a}^S)^\dagger\right)$
be a quadrature amplitude, with PVM
$\mathbb{P}^{S\to \R}$.
Suppose that this PVM is implemented as an indirect measurement
using other optical fields as ancillary systems, 
where the interaction preserves
the total photon number and the final measurement is obtained by
post-processing the photon-counting outcomes.
This measurement model is covariant with respect to the $U(1)$ symmetry
generated by the total photon number, i.e., 
\begin{equation}
    V_\theta^X=e^{i\theta \hat{n}^X},\qquad
    T_\theta^\R=\id_\R,
\end{equation}

where
$\hat{n}^X=(\hat{a}^X)^\dagger\hat{a}^X$
is the number operator of system $X$.
The rotated PVM
$\mathbb{P}^{S\to \R}_\bullet[\theta]
\coloneqq
\mathcal{V}^S_{\theta}(\mathbb{P}^{S\to \R}_\bullet)$
is the PVM of
$\hat{q}^S(\theta)
    \coloneqq
    \mathcal{V}^S_{\theta}(\hat{q}^S)
    =
    \frac{1}{2}
    \left(
        e^{-i\theta}\hat{a}^S
        +
        e^{i\theta}(\hat{a}^S)^\dagger
    \right).$
Using coherent states $\ket{\alpha}$, 
one finds that, 
for every $\theta\in(0,2\pi)$,
\begin{align}
    \norm{
    \mathbb{P}^{S\to \R}_{[0,\infty)}[\theta]
    -
    \mathbb{P}^{S\to \R}_{[0,\infty)}
    }
    &\ge
    \sup_{\alpha\in\C}
    \abs*{
    \braket[3]{\alpha}{
    \mathbb{P}^{S\to \R}_{[0,\infty)}[\theta]
    -
    \mathbb{P}^{S\to \R}_{[0,\infty)}
    }{\alpha}
    }\ge 1.
    \label{eq:orthogonalize_coherent}
\end{align}
The second inequality is shown as follows. 
For a coherent state $\ket{\alpha}$, the quadrature
$\hat{q}^S(\theta)$ has a Gaussian distribution with mean
$\Re(e^{-i\theta}\alpha)$ and variance $1/4$.
Therefore,
\begin{equation}
    \lim_{R\to\infty}
    \bra{Re^{i\phi}}
        \chi_{[0,\infty)}\bigl(\hat{q}^S(\theta)\bigr)
    \ket{Re^{i\phi}}
    =
    \begin{cases}
        1, & \cos(\phi-\theta)>0,\\
        0, & \cos(\phi-\theta)<0 .
    \end{cases}
\end{equation}
Hence, for any fixed $\theta\in(0,2\pi)$, we have
\begin{align}
    &\sup_{\alpha\in\C}
    \abs*{\braket[3]{\alpha}{\chi_{[0,\infty)}(\hat{q}^S)-\chi_{[0,\infty)}(\hat{q}^S(\theta))}{\alpha}}\nonumber\\
    &\ge 
    \lim_{R\to\infty}
    \abs*{\braket[3]{Re^{i\frac{\theta-\pi}{2}}}{\chi_{[0,\infty)}(\hat{q}^S)-\chi_{[0,\infty)}(\hat{q}^S(\theta))}{Re^{i\frac{\theta-\pi}{2}}}}\nonumber\\
    &=1.
\end{align}

Equation~\eqref{eq:orthogonalize_coherent} is the optical analogue of the
orthogonalization property used for the position measurement in the main text.
Indeed, the denominator in Eq.~\eqref{eq:main_PVM} is equal to one for every
$\theta\in(0,2\pi)$. 
On the other hand,
\begin{equation}
    F\bigl(\eta^A,\mathcal{V}_\theta^A(\eta^A)\bigr)\to1
    \qquad
    (\theta\to0)
\end{equation}
for every ancilla state $\eta^A$.
Hence Eq.~\eqref{eq:main_PVM} can hold only if
\begin{equation}
    \epsilon\ge \frac{1}{2}.
\end{equation}
Thus no approximate implementation can be uniformly accurate over all input states.

\section{General results for restricted input state sets}

The results in the main text also extend to the case where the implementation
error is evaluated only on a specified set of input states.
We state the extension only for measurements; the case of
quantum channels is completely analogous.

Let $\mathcal{S}$ be a set of density operators.
For two measurements $\mathbb{P}$ and $\mathbb{Q}$, we define their distance on
$\mathcal{S}$ by
\begin{equation}
    P_{\mathcal{S}}(\mathbb{P},\mathbb{Q})
    \coloneqq
    \sup_{\rho\in\mathcal{S}}
    \left[
        1-
        F_{\mathrm{BC}}\bigl(\mathbb{P}(\rho),\mathbb{Q}(\rho)\bigr)^2
    \right]^{1/2}.
\end{equation}
We say that an ancilla state $\eta^A$ and a measurement
$\mathbb{M}^{SA\to\Omega}$ implement a measurement
$\mathbb{P}^{S\to\Omega}$ on $\mathcal{S}$ with error $\epsilon\ge 0$ if
\begin{equation}
    P_{\mathcal{S}}
    \bigl(
        \mathbb{P}^{S\to\Omega},
        \mathbb{M}^{SA\to\Omega}(\bullet\otimes\eta^A)
    \bigr)
    \le
    \epsilon .
    \label{eq:approx_implement_restricted}
\end{equation}
For two sets of states $\mathcal{S}_1$ and $\mathcal{S}_2$, define
\begin{align}
    \hat{F}_{\mathcal{S}_1,\mathcal{S}_2}
    (\lambda;\mathbb{P},\mathbb{Q})
    &\coloneqq
    \inf_{\substack{
        (\rho,\sigma)\in\mathcal{S}_1\times\mathcal{S}_2\\
        F(\rho,\sigma)\ge \lambda
    }}
    F_{\mathrm{BC}}\bigl(\mathbb{P}(\rho),\mathbb{Q}(\sigma)\bigr),
    \\
    d_{\mathcal{S}_1,\mathcal{S}_2}(\mathbb{P},\mathbb{Q})
    &\coloneqq
    \sup\Bigl\{
        \lambda\in[0,1]\ \Big|\ 
        \hat{F}_{\mathcal{S}_1,\mathcal{S}_2}
        (\lambda;\mathbb{P},\mathbb{Q})=0
    \Bigr\}.
\end{align}
Then Theorem~\ref{thm:NP_measurement} admits the following direct
generalization.

\begin{thm}
Suppose that, for each $i=1,2$, the pair
$(\mathbb{M}^{SA\to\Omega},\eta_i^A)$ implements
$\mathbb{P}_i^{S\to\Omega}$ on $\mathcal{S}_i$ with error $\epsilon$.
Then
\begin{equation}
    F(\eta_1^A,\eta_2^A)
    \le
    \inf_{\lambda>0}
    \frac{
        \hat{F}_{\mathcal{S}_1,\mathcal{S}_2}
        (\lambda;\mathbb{P}_1^{S\to\Omega},\mathbb{P}_2^{S\to\Omega})
        +2\epsilon
    }{\lambda}.
\end{equation}
In particular, if
$d_{\mathcal{S}_1,\mathcal{S}_2}
(\mathbb{P}_1^{S\to\Omega},\mathbb{P}_2^{S\to\Omega})>0$, then
\begin{equation}
    F(\eta_1^A,\eta_2^A)
    \le
    \frac{
        2\epsilon
    }{
        d_{\mathcal{S}_1,\mathcal{S}_2}
        (\mathbb{P}_1^{S\to\Omega},\mathbb{P}_2^{S\to\Omega})
    } .
\end{equation}
\end{thm}

Applying this restricted version yields the corresponding WAY-type statement
for restricted input state sets.
Assume that $\mathbb{M}^{SA\to\Omega}$ in
Eq.~\eqref{eq:approx_implement_restricted} is covariant.
Then
\begin{equation}
P_{\mathcal{V}_g^S(\mathcal{S})}(\mathcal{T}_g^\Omega\circ\mathbb{P}^{S\sto  \Omega}\circ\mathcal{V}^S_{g^{-1}},\mathbb{M}^{SA\sto\Omega}(\bullet\otimes \mathcal{V}^A_g(\eta^A)))\le \epsilon.
\end{equation}
Accordingly, define
\begin{align}
    &\hat{F}_{\mathcal{S}}
    (\lambda;\mathbb{P}^{S\to\Omega},g)
    \nonumber\\
    &\quad\coloneqq
    \hat{F}_{\mathcal{S},\mathcal{V}_g^S(\mathcal{S})}
    \bigl(
        \lambda;
        \mathbb{P}^{S\to\Omega},
        \mathcal{T}_g^\Omega
        \circ
        \mathbb{P}^{S\to\Omega}
        \circ
        \mathcal{V}_{g^{-1}}^S
    \bigr)
    \nonumber\\
    &\quad=
    \inf_{\substack{
        (\rho,\sigma)\in\mathcal{S}^2\\
        F\bigl(\rho,\mathcal{V}_g^S(\sigma)\bigr)\ge \lambda
    }}
    F_{\mathrm{BC}}
    \bigl(
        \mathbb{P}^{S\to\Omega}(\rho),
        \mathcal{T}_g^\Omega
        \circ
        \mathbb{P}^{S\to\Omega}(\sigma)
    \bigr).
\end{align}
We then obtain the following restricted-input version of the WAY-type bound.

\begin{thm}
If an ancilla state $\eta^A$ and a covariant measurement
$\mathbb{M}^{SA\to\Omega}$ implement a measurement $\mathbb{P}^{S\to\Omega}$ on
$\mathcal{S}$ with error $\epsilon$, then, for every $g\in G$,
\begin{equation}
    F\bigl(\eta^A,\mathcal{V}_g^A(\eta^A)\bigr)
    \le
    \inf_{\lambda>0}
    \frac{
        \hat{F}_{\mathcal{S}}
        (\lambda;\mathbb{P}^{S\to\Omega},g)
        +2\epsilon
    }{\lambda}.
\end{equation}
\end{thm}

\bibliographystyleSM{apsrev4-2-titles-revised}
\bibliographySM{bibliographySM}

\end{document}